\documentclass[a4paper,UKenglish,cleveref, autoref, thm-restate]{lipics-v2021}
\usepackage[linesnumbered,ruled,vlined]{algorithm2e}
\usepackage{algpseudocode}

\graphicspath{ {./Figures/} }
\nolinenumbers
\title{Succinct Data Structure for Chordal Graphs with Bounded Vertex Leafage} 

\author{Girish Balakrishnan}{Indian Institute of Technology Madras, Chennai, India}{girishb@cse.iitm.ac.in}{https://orcid.org/0000-0002-1825-0097}{}
\author{Sankardeep Chakraborty}{University of Tokyo, Tokyo, Japan}{sankardeep.chakraborty@gmail.com}{}{}
\author{N S Narayanaswamy}{Indian Institute of Technology Madras, Chennai, India}{swamy@cse.iitm.ac.in}{}{}
\author{Kunihiko Sadakane}{University of Tokyo, Tokyo, Japan}{sada@mist.i.u-tokyo.ac.jp}{}{}

\authorrunning{Balakrishnan, Chakraborty, Narayanaswamy, and Sadakane} 

\Copyright{Jane Open Access and Joan R. Public} 



\category{} 

\relatedversion{} 

\EventEditors{John Q. Open and Joan R. Access}
\EventNoEds{2}
\EventLongTitle{42nd Conference on Very Important Topics (CVIT 2016)}
\EventShortTitle{CVIT 2016}
\EventAcronym{CVIT}
\EventYear{2016}
\EventDate{December 24--27, 2016}
\EventLocation{Little Whinging, United Kingdom}
\EventLogo{}
\SeriesVolume{42}
\ArticleNo{23}

\begin{document}

\maketitle
\begin{abstract}
Chordal graphs is a well-studied large graph class that is also a strict super-class of path graphs. Munro and Wu (ISAAC 2018) have given an $(n^2/4+o(n^2))-$bit succinct representation for $n-$vertex unlabeled chordal graphs. A chordal graph $G=(V,E)$ is the intersection graph of sub-trees of a tree $T$.  Based on this characterization, the two parameters of chordal graphs which we consider in this work are \textit{leafage}, introduced by Lin, McKee and West (Discussiones Mathematicae Graph Theory 1998) and \textit{vertex leafage}, introduced by Chaplick and Stacho (Discret. Appl. Math. 2014). Leafage is the minimum number of leaves in any possible tree $T$ characterizing $G$. Let $L(u)$ denote the number of leaves of the sub-tree in $T$ corresponding to $u \in V$ and $k=\max\limits_{u \in V} L(u)$. The smallest $k$ for which there exists a tree $T$ for $G$ is called its vertex leafage.

In this work, we improve the worst-case information theoretic lower bound of Munro and Wu (ISAAC 2018) for $n-$vertex unlabeled chordal graphs when vertex leafage is bounded and leafage is unbounded.  The class of unlabeled $k-$vertex leafage chordal graphs that consists of all chordal graphs with vertex leafage at most $k$ and unbounded leafage, denoted $\mathcal{G}_k$, is introduced for the first time.  For $k>0$ in  $o(n/\log n)$, we obtain a lower bound of $((k-1)n \log n -kn \log k - O(\log n))-$bits on the size of any data structure that encodes a graph in $\mathcal{G}_k$. Further, for every $k-$vertex leafage chordal graph $G$ such that for $k>1$ in  $o(n^c), c >0$, we present a $((k-1)n  \log n + o(kn \log n))-$bit succinct data structure, constructed using the succinct data structure for path graphs with $kn/2$ vertices. Our data structure supports adjacency query in $O(k \log n)$ time and using additional $2n \log n$ bits, an $O(k^2 d_v \log n + \log^2 n)$ time neighbourhood query where $d_v$ is degree of $v \in V$. 
\end{abstract}

\section{Introduction}
A data structure for a graph class $\mathcal{G}$ of graphs with $n-$vertices is succinct if it takes $(\log |\mathcal{G}|+o(\log|\mathcal{G}|))-$bits of space; here $|\mathcal{G}|$ denotes the number of graphs in $\mathcal{G}$. A succinct representation for graph $G=(V,E)$ in $\mathcal{G}$ is expected to support the following queries for each pair of vertices $u,v \in V$:
\begin{itemize}
\itemsep0em
\item $\texttt{adjacency}(u, v)$: returns "YES" if and only if vertices $u$ and $v$ are adjacent in $G$.
\item $\texttt{neighborhood}(u)$: returns all the vertices in $V$ that are adjacent to vertex $u$.
\item $\texttt{degree}(u)$: returns the number of vertices adjacent to vertex $u$.
\end{itemize}

\noindent
The earliest work on space-efficient data structures for graph classes was by Itai and Rodeh~\cite{IR} for vertex labeled planar graphs and by Jacobson~\cite{Jacobson1989} for class of static unlabeled trees and planar graphs. In recent years, succinct data structures for intersection graphs is getting lot of attention. A few of them are, Acan et al.~\cite{HSSS} for interval graphs, Acan et al.~\cite{HSSKKS2020} for families of intersection graphs on circle, Munro and Wu~\cite{Munro_Wu} for chordal graphs and Balakrishnan et al.~\cite{BCNS2024} for path graphs. Also, Markenzon et al.~\cite{MARKENZON2013331} gives an efficient representation for chordal graphs.

A recent paper by Chakraborty and Jo~\cite{SS} improve the information theoretic lower bound for interval graphs as given by Acan et al.~\cite{HSSS} by bounding maximum degree. For bounded maximum degree $\Delta$ in $O(n^{\epsilon})$ where $1 < \epsilon < 1$, they give a lower bound of $(\frac{1}{6}n\log \Delta - O(n))-$bits and a $(n \log \Delta + O(n))-$bit space-efficient data structure. They also give a $((\chi -1)n + o(\chi n))-$bit space-efficient data structure for interval graphs with bounded chromatic number $\chi$ for $\chi=o(\log n)$. In Balakrishnan et al.~\cite{GSNK2023Arxiv}, such a parameterization has been applied to a larger graph class, namely, the class of graphs  with $d-$dimensional $t-$representation where parameters $d$ and $t$ are bounded. A $((2td-1)n\log n + o(tdn\log n))-$bit succinct data structure for graphs with $d-$dimensional $t-$representation is presented in~\cite{GSNK2023Arxiv}. Special cases of this graph class are graphs with bounded boxicity with $t=1$ and bounded interval number with $d=1$. Chordal graphs is also a large graph class that is a generalization of interval graphs. A graph is a chordal graph if it contains only cycles of length at most three. Lin et al.~\cite{LMW} introduced the parameter leafage and Chaplick and Stacho~\cite{CS} the parameter vertex leafage for chordal graphs. In this work, we define the $k-$vertex leafage chordal graphs that contains chordal graphs with bounded vertex leafage and unbounded leafage. When vertex leafage and leafage are equal to two we get interval graphs and when vertex leafage is two and leafage is unbounded we get path graphs. We present a data structure for $k-$vertex leafage chordal graphs using succinct data structure for path graphs and prove that it is succinct.

\noindent
{\it Convention.} Throughout the rest of this paper, set of vertices and edges of a graph $G$ will be denoted by $V(G)$ and $E(G)$, respectively.

\noindent
{\bf Our Results.} We present the following two theorems. The first theorem proves a lower bound on the class of $k-$vertex leafage graphs, denoted $\mathcal{G}_k$.

\begin{restatable}[]{theorem}{kleafagelb}
\label{thm:kleafagelb}
For $k >0$ in $o(n^c), c > 0, \log |\mathcal{G}_k| \ge (k-1)n \log n -kn \log k - O(\log n)$.    
\end{restatable}

\noindent
The next theorem proves the existence of a matching  data structure that supports adjacency query efficiently. For neighbour query we use additional $2n \log n$ bits.

\begin{restatable}[]{theorem}{succinctds}
\label{thm:succinctds}
For $k>1$ and in $o(n^c), c >0$, a graph $G \in \mathcal{G}_k$ has a $(k-1)n \log n + o(kn \log n)$-bit succinct data structure that supports adjacency query in $O(k\log n)$ time and using additional $2n\log n$ bits the neighbourhood query for vertex $v$ in $O(k^2 d_v \log  n + \log^2 n)$ time where $d_v$ is the degree of $v \in V(G)$.    
\end{restatable}

\noindent
{\bf Our Main Ideas.} The main ideas behind this work are two fold.
\begin{enumerate}
    \itemsep0em
    \item Demonstrate an improved worst-case information theoretic lower bound for chordal graphs with bounded vertex leafage by using a simple and constructive counting technique motivated by the method of partial coloring as used by Acan et al.~\cite{HSSKKS2020}.
    \item By carefully decomposing the sub-trees of a tree into paths, present a data structure for chordal graphs with bounded vertex leafage that uses the succinct data structure for path graphs given in Balakrishnan et al.~\cite{BCNS2024} as a black box. Also, using the above improved lower bound show that this data structure is succinct.
\end{enumerate}

\noindent
{\bf Organisation of the Paper.} Section~\ref{sec:prelims} gives details of the compact data structures we have used in the construction of the succinct data structure of $k-$vertex leafage chordal graphs along with characterisation of chordal and path graphs. It also formalizes and explains the method of partial coloring, first used in Acan et al.~\cite{HSSKKS2020} to obtain the lower bound for $\mathcal{G}_k$. Section~\ref{sec:k-leafage} defines the $k-$vertex leafage chordal graphs and also gives a lower bound on the cardinality of the class. The succinct data structure design is given in Section~\ref{sec:succrep} and it contains in Section~\ref{sec:transformation}, the method to convert a $k-$vertex leafage chordal graph with $n$ vertices to a path graph with $kn/2$ vertices. 
Section~\ref{sec:conclusion} concludes by giving motivation to extend this work by generalizing the parameter vertex leafage and leafage to general graphs.

\section{Preliminaries}
\label{sec:prelims}

From Gavril~\cite{Gavril} we have, the following characterisation of chordal graphs; see Golumbic~\cite{agtpg} for more details on chordal graphs or otherwise called triangulated graphs.
\begin{theorem}
\label{thm:chordal}    
The graph $G$ is a chordal graph if and only if there exists a clique tree $T$, such that for every $v \in V(G)$, the set of maximal cliques containing $v$ form a sub-tree of $T$ such that two vertices are adjacent if the corresponding sub-trees intersect.
\end{theorem}

\noindent
Also, from Gavril~\cite{Gavril_path} we have the following characterisation of path graphs; see Monma and Wu~\cite{MonmaW86} for more details on path graphs.

\begin{theorem}
\label{thm:pg}
The graph $G$ is a path graph if and only if  there exists a clique tree $T$, such that for every $v \in V(G)$, the set of maximal cliques containing $v$ is a path in $T$ such that two vertices are adjacent if the corresponding paths intersect.
\end{theorem}

\noindent
{\bf Succinct Data Structure for Ordinal Trees.} Tree $T$ is called an \textit{ordinal tree} if for $z>0$ and any $u \in V(T)$ with children $\{u_1,\ldots,u_z\}$, for $1 \le i < j \le z$, $u_i$ is to the left of  $u_j$~\cite{Nav_Sada}.  By considering ordinal trees as balanced parenthesis Navarro and Sadakane~\cite{Nav_Sada} has given a $2n+o(n)$ bit succinct data structure. 
\begin{lemma}
\label{lem:2ntree}
For any ordinal tree $T$ with $n$ nodes, there exists a $2n+o(n)$ bit Balanced Parentheses (BP) based data structure that supports the following functions among others in constant time :
\begin{enumerate}
    \itemsep0em 
    \item ${\normalfont\texttt{lca}}(i,j)$, returns the lowest common ancestor of two nodes $i,j$ in $T$,  and
    \item ${\normalfont\texttt{child}}(i,v)$, returns the $q-$th  child of $v$ in $T$.
\end{enumerate}
\end{lemma}

\noindent
{\bf Rank-Select Data Structure.} Bit-vectors are extensively used in the succinct representation given in Section~\ref{sec:succrep}. The following data structure due to Golynski et al. \cite{GMS2006} and the functions supported by  it are useful.

\begin{lemma}
\label{lem:bstr}
Let $B$ be an $n-$bit vector and  $b \in \{0,1\}$.  There exists an $n+o(n)$ bit data structure that supports the following functions in constant time:
\begin{enumerate}
    \itemsep0em 
    \item ${\normalfont\texttt{rank}}(B,b,i)$:  Returns the number of $b$'s up to and including position $i$ in the bit vector $B$ from the left. 
    \item ${\normalfont\texttt{select}}(B,b,i)$: Returns the position of the $i$-th $b$ in the bit vector $B$ from left. For $i \notin [n]$ it returns 0.
\end{enumerate}
\end{lemma}

\noindent
{\bf Non-Decreasing Integer Sequence Data Structure.} Given a set of positive integers in the non-decreasing order we can store them efficiently using the differential encoding scheme for increasing numbers; see Section 2.8 of~\cite{Navarro}. Let $S$ be the data structure that supports differential encoding for increasing numbers then the function $\texttt{accessNS}(S,i)$ returns the $i-$th number in the sequence. 

\begin{lemma}
\label{lem:ssn}
Let $S$ be a sequence of $n$ non-decreasing positive integers $a_1,\ldots,a_n, 1 \leq a_i \leq n$. There exists a $2n+o(n)$ bit data structure that supports ${\normalfont\texttt{accessNS}}(S,i)$ in constant time.
\end{lemma}
\begin{proof}
We will prove the lemma by giving a construction of such a data structure. $a_1$ will be represented by a sequence of $a_1$ 1's followed by a 0. Subsequently, $a_i$'s are represented by storing $a_i-a_{i-1}$ many 1's followed by a 0. It will take at most $2n$ bits since there are $n$ 0's and at most $n$ 1's. Let this bit string be stored using the data structure of Lemma~\ref{lem:bstr} and be denoted as $B$. $B$ takes $2n+o(n)$ bits. $\texttt{accessNS}(S,i)$ can be implemented using $\texttt{rank}(B,1,\texttt{select}(B,0,i))$.
\end{proof}

\noindent
{\bf Succinct Data Structure for Path Graphs.} From~\cite{BCNS2024} we have the following succinct data structure for path graphs that supports adjacency and neighbourhood queries with slight modifications in input. In the following lemma the endpoints of paths  are input to the queries whereas in~\cite{BCNS2024} the path  indices are given as input.
\begin{lemma}
\label{lem:succpg}
A path graph $G$ has an $n \log n + o(n\log n)$-bit succinct representation. For a $u \in V(G)$, let $P_u=(s_u,t_u)$ be the path corresponding $u$ in clique tree $T$ of $G$. The succinct representation constructed from the clique tree $T$ representation supports for $u \in V(G)$ the following  queries:
\begin{enumerate}
    \itemsep0em 
    \item ${\normalfont \texttt{adjacencyPG}(s_u, t_u, s_v, t_v)}$: Returns true if $P_u=(s_u,t_u)$ intersects $P_v=(s_v,t_v)$ in $O(\log n)$ time else false,
    \item ${\normalfont \texttt{pathep}(u)}$:  Returns the endpoints of path $P_u$, corresponding to $u$, in $T$ in $O(\log  n)$ time, and
    \item ${\normalfont \texttt{neighbourhoodPG}(s_u, t_u)}$: Returns the list of paths intersecting $P_u=(s_u,t_u)$ in $O(d_u \log n)$ time where $d_u$ is the degree of vertex $u$.
    \item ${\normalfont \texttt{getHPStartNode}(v)}$: Returns the start node of heavy path $\pi$ that contains $v \in V(T)$ in constant time. If $v$ is not the first child, that is, it is adjacent to its parent by a light edge, then $v$ itself is returned.
\end{enumerate}
\end{lemma}

\noindent{\bf Permutations.} The following  data structure by  Munro et al.~\cite{DBLP:journals/tcs/MunroRRR12}, gives a  succinct representation for storing permutation of $[n]$.
\begin{lemma}[\cite{DBLP:journals/tcs/MunroRRR12}]\label{lem:perm}
Given a permutation of $[n]$ there exists an $(n\log n+ o(n\log n))$-bit data structure that supports the following queries.
\begin{itemize}
    \itemsep0em
    \item $\pi(i)$: Returns the $i-$th value in the permutation in $O(1)$ time.
    \item $\pi^{-1}(j)$: Returns the position of the $j-$th value in the permutation in $O(f(n))$ time for any increasing function $f(n)=o(\log n)$.
\end{itemize}
\end{lemma}

\noindent
{\bf Method of Partial Coloring.} Let $\mathcal{G}$ be a graph class. A \textit{partial coloring} of $G \in \mathcal{G}$ is the triple $\langle G, U, g \rangle$ where,
\begin{itemize}
    \itemsep0em
    \item $U \subseteq V(G)$ such that for $s>0, |U|=s$, and
    \item $g:U \rightarrow\{1,\ldots,s\}$ is a bijection.
\end{itemize}

\noindent
Every vertex $u \in U$ is said to have color $g(u)$ and vertices in $V(G)\backslash U$ are said to be uncolored. Two partially colored graphs $\langle H_1,U_1,g_1 \rangle$ and $\langle H_2,U_2,g_2 \rangle$ are said to be different when either:
\begin{enumerate}
    \itemsep0em
    \item $E(H_1) \ne E(H_2)$, or
    \item $E(H_1)=E(H_2)=E(H)$ for some $H \in \mathcal{G}$ and 
    \begin{enumerate}
        \itemsep0em
         \item $U_1 \ne U_2$ and there does not exist a bijection $f:V(H_1) \rightarrow V(H_2)$ such that for every $u \in V(H_1)\backslash U_1$ there exists a $f(u) \in V(H_2)\backslash U_2$ with same set of colored neighbours, or
         \item for $U_1=U_2=U$ there exists $u \in V(H)\backslash U$ such that its colored neighbourhood in $\langle H_1,U_1,g_1 \rangle$ and $\langle H_2,U_2,g_2 \rangle$ are different.
    \end{enumerate}
\end{enumerate}
Else, they are same. The method of counting by partial coloring as given in Theorem 1 of Acan et al.~\cite{HSSKKS2020} can be defined using the following proposition.
\begin{proposition}
\label{prop:partialcoloring}
    Let $\mathcal{G}'$ be the class of partially colored graphs obtained from class of graphs $\mathcal{G}$ by selecting $m$ vertices out of $n$ and coloring them using $m$ distinct colors. Then, $|\mathcal{G}'| \le {n \choose m} m! |\mathcal{G}|$. If there exists a graph class $\mathcal{G}^c \subset \mathcal{G}'$ then $|\mathcal{G}'| \ge |\mathcal{G}^c|$ and $|\mathcal{G}| \ge \frac{|\mathcal{G}^c|}{{n \choose m} m!}$.
\end{proposition}

\noindent
{\it Remark:} While computing $|\mathcal{G}'|$, indistinguishable partially colored graphs can also be counted since we only require an upper bound, however, this is not the case while computing $|\mathcal{G}^c|$. 

\section{Class of $k-$Vertex Leafage Chordal Graphs and its Lower Bound}
\label{sec:k-leafage}
In this section, we define the class of $k-$vertex leafage chordal graphs, denoted $\mathcal{G}_k$, followed by the lower bound for $\log|\mathcal{G}_k|$. According to Theorem~\ref{thm:chordal}, a graph $G$ is a chordal graph if there exists a tree $T$ such that corresponding to every vertex $v \in V(G)$ there exists a sub-tree $T_v$ of $T$ and $\{u,v\} \in E(G)$ if and only if $V(T_u)  \cap V(T_v) \ne \phi$ where $T_u$ is the sub-tree corresponding to $u$. We call $T$ the tree model of $G$. For every chordal graph there exists a tree $T$ called the \textit{clique tree} such that $V(T)$ is the set of maximal cliques of $G$. The \textit{leafage} of a chordal graph $G$, denoted $l(G)$, is the minimum number of leaves of a tree $T$ out of all possible trees characterising $G$. Leafage was studied by Lin et al. in~\cite{LMW}. Later, Chaplick and Stacho in~\cite{CS} studied \textit{vertex leafage}, denoted $vl(G)$. Let $L(u)$ denote the number of leaves of the sub-tree in $T$ corresponding to $u \in V$ and $k=\max\limits_{u \in V} L(u)$. The smallest $k$ for which there exists a tree $T$ for $G$ is called its vertex leafage.

\noindent
{\bf Class of $k-$Vertex Leafage Chordal Graphs.} The class of chordal graphs can be considered as the generalisation of path graphs using  vertex leafage as the parameter. Theorem~\ref{thm:pg} implies that path graphs are chordal graphs with vertex leafage equal to two and unbounded leafage. Generalizing this, $\mathcal{G}_k$ is the set of all chordal graphs with vertex leafage at most $k$ and unbounded leafage. Succinct data structure for path graphs is given by Balakrishnan et al. in~\cite{BCNS2024}. In this paper, we present a succinct data structure for chordal graphs with vertex leafage at most $k$ for $k \in o(n^c), c >0$ using succinct data structure for path graphs as black-box.

\noindent
{\bf Lower Bound.} Counting chordal graphs involves heavy  mathematical machinery as can be seen from Wormald~\cite{Wormald1985} on counting labeled chordal graphs which is used by Munro and Wu~\cite{Munro_Wu} to obtain lower bound for unlabeled chordal graphs. Here we give a much simpler technique for counting unlabeled chordal graphs with bounded vertex leafage. In order to derive the lower bound for $\log |\mathcal{G}_k|$ by implementing Proposition~\ref{prop:partialcoloring}, we define two new graph classes, $\mathcal{G}'_k$ and $\mathcal{G}^c_k$, where $\mathcal{G}'_k$ corresponds to $\mathcal{G}'$ and $\mathcal{G}^c_k$ to $\mathcal{G}^c$ of Proposition~\ref{prop:partialcoloring}.

\noindent
{\bf Graph Class $\mathcal{G}'_k$.} We consider the class of all partially colored $k-$vertex leafage chordal graphs, denoted $\mathcal{G}'_k$, that has for fixed $1 \le m \le n$, $m$ out of $n$ vertices colored using colors $\{1,\ldots,m\}$. Graphs in $\mathcal{G}'_k$ are obtained from graphs in $\mathcal{G}_k$ as follows. The input to the  procedure is $G \in \mathcal{G}_k$ and a set $\{v_1,\ldots,v_m\}$ of $m$ vertices of $G$. For each $G \in \mathcal{G}_k$, we get a set of ${n \choose m}m!$ graphs of $\mathcal{G}'_k$ where each graph $G'$ is obtained by coloring the selected $m$  vertices of $G$ by a permutation of $\{1,\ldots,m\}$. A partially colored graph in $\mathcal{G}'_k$ is denoted $\langle H, U, g \rangle$ where $U \subset V(H), |U|=m$, and $g: U \rightarrow \{1,\ldots,m\}$.
\begin{lemma}
\label{lem:colorGprime1}
For each $k \geq 1$, $\log |\mathcal{G}'_k| \le \log |\mathcal{G}_k| + n \log n - (n-m)\log(n-m) - 2n + O(\log n)$.
\end{lemma}
\begin{proof}
The $m$ vertices given as input to the procedure can be selected in $n \choose m$ ways and there are $m!$ ways of coloring it. Thus, from each $G$ in $\mathcal{G}_k$ we get ${n \choose m} m!$ partially colored graphs. Hence, we have $|\mathcal{G}'_k| \le |\mathcal{G}_k| \frac{n!}{(n-m)!}$. Taking log on both sides and using Stirling's approximation, that is, $\log n!=n \log n - n + O(\log n)$, we get $\log |\mathcal{G}'_k| \le \log |\mathcal{G}_k| + n \log n - (n-m)\log(n-m) - 2n + O(\log n)$. 
\end{proof}

\noindent
{\bf Graph Class $\mathcal{G}^c_k$.} As per Proposition~\ref{prop:partialcoloring}, we construct the class $\mathcal{G}^c_k \subset \mathcal{G}'_k$ for which we can obtain exact count or a lower bound.  We give a construction mechanism for graphs in $\mathcal{G}^c_k$ such that all graphs in it have the following properties. For $G \in \mathcal{G}^c_k$,
\begin{itemize}
    \itemsep0em
    \item $U \subseteq V(G)$, with $|U|=m$, is fixed and have a fixed coloring using colors $\{1,\ldots,m\}$,
    \item let $T$ be the tree corresponding to $G$ then the sub-trees in the clique tree corresponding to  vertices in $U$ consist of only one node, and
    \item sub-trees of $T$ corresponding to vertices in $V(G) \backslash U$ have at most $k$ leaves with at least $\frac{m+1}{2(k-1)}$ sub-trees with $k$ leaves.
\end{itemize}

\noindent
In other words, for all the partially colored graphs in $\mathcal{G}^c_k$, $U$ and $g$ are fixed. Based on the tree model, the vertices of a graph $H \in \mathcal{G}^c_k$ with tree model $T$ are of two types:
\begin{enumerate}
    \itemsep0em
    \item {\bf basis vertices $U$}: $m$ vertices of $U$ are represented by single-node sub-tree in $T$, and
    \item {\bf dependent vertices $V(H) \backslash U$}: rest of the $(n-m)$ vertices are represented by sub-trees in $T$ with number of leaves at most $k$ with at least at least $\frac{m+1}{2(k-1)}$ sub-trees with $k$ leaves. Let the dependent vertices corresponding to these  $\frac{m+1}{2(k-1)}$ sub-trees be denoted $U'$. 
\end{enumerate}

\noindent
We have the following useful proposition:
\begin{proposition}
    \label{prop:uniquetree}
    For a rooted tree $T$, the set $V' \subseteq V(T)$ uniquely defines a sub-tree $T'$ of $T$ such that if $u,v \in V'$ then the path connecting it is in $T'$.
\end{proposition}

\noindent
The sub-trees corresponding to the basis and dependent vertices are called basis and dependent sub-trees, respectively. Let $F=\frac{m+1}{2(k-1)}$. The input to the procedure that constructs $H$ are:
\begin{enumerate}
    \itemsep0em
    \item $n, m, k,$ and
    \item $\mathcal{J}=\{J_1,\ldots,J_{n-m-F}\}$ where for $1 \le i \le n-m-F, 1 \le t \le k$ , $J_i=\{a^i_1,\ldots,a^i_k\}$ and $1 \le a^i_t \le m$.
\end{enumerate}
The construction mechanism that constructs the partially colored graph $\langle H, U, g \rangle$ is as follows. 
\begin{enumerate}
    \itemsep0em
    \item Consider a rooted complete binary tree $T$ with $m>0$ nodes and let them be colored from 1 to $m$.
    \item {\bf Construction of basis sub-trees.} For $1 \le j \le m$, let each node $b_j \in V(T)$ be a single node sub-tree, $T_j$. These sub-trees are the \textit{basis sub-trees}. The basis sub-trees correspond to the $m$ basis vertices of $H$, denoted by $U$. Let the basis vertices be colored by the color assigned to the nodes to which they are assigned.
    \item {\bf Construction of dependent sub-trees.} First we define the fixed sub-trees with $k$ leaves. Since $T$ is a complete binary tree it has $\frac{m+1}{2}$ leaves. Let the set of leaves of $T$ be denoted by $L$. For $1 \le t \le F$, partition $L$ into blocks $L=\bigcup\limits_t L_t$ such that $|L_t|=k-1$ and for $1 \le z < z' \le t$, there does not exist a node in $L_z$ with color greater the smallest color of nodes in $L_{z'}$. For every $u_t \in U'$, construct a sub-tree $T_{u_t}$ in $T$ by connecting paths from the $k-1$ leaves in $L_t$ to the root of $T$. This ensures that $T_{u_t}$ has $k$ leaves. So, $|U'|=F$. For $1 \le i \le n-m-F$, construct sub-tree $T_i$ from $J_i=\{a^i_1,\ldots,a^i_k\}$, where for $1 \le t \le k, a^i_t \in V(T)$, such that $J_i$ is the set of $k$ nodes of $T_i$; as per Proposition~\ref{prop:uniquetree}, $J_i$ uniquely defines a sub-tree. These sub-trees are called \textit{dependent sub-trees} and correspond to the dependent vertices in $V(H) \backslash U$. The sub-trees corresponding to nodes of $U'$ ensure that the chordal graph is connected and there are at least $F$ sub-trees with $k$ leaves there by making  $H$ a $k-$vertex leafage chordal graph. Also, apart from $U$ and $g$, $U'$ is also fixed for any partially colored $k-$vertex leafage chordal graph in $\mathcal{G}^c_k$.
\end{enumerate}

\noindent
{\it Convention.} A node $a_j \in V(T)$ will also be used to denote the sub-tree $T_j$ corresponding to the basis vertex of $H$.

\begin{figure}[ht]
\centering
\includegraphics[width=12cm,height=6cm]{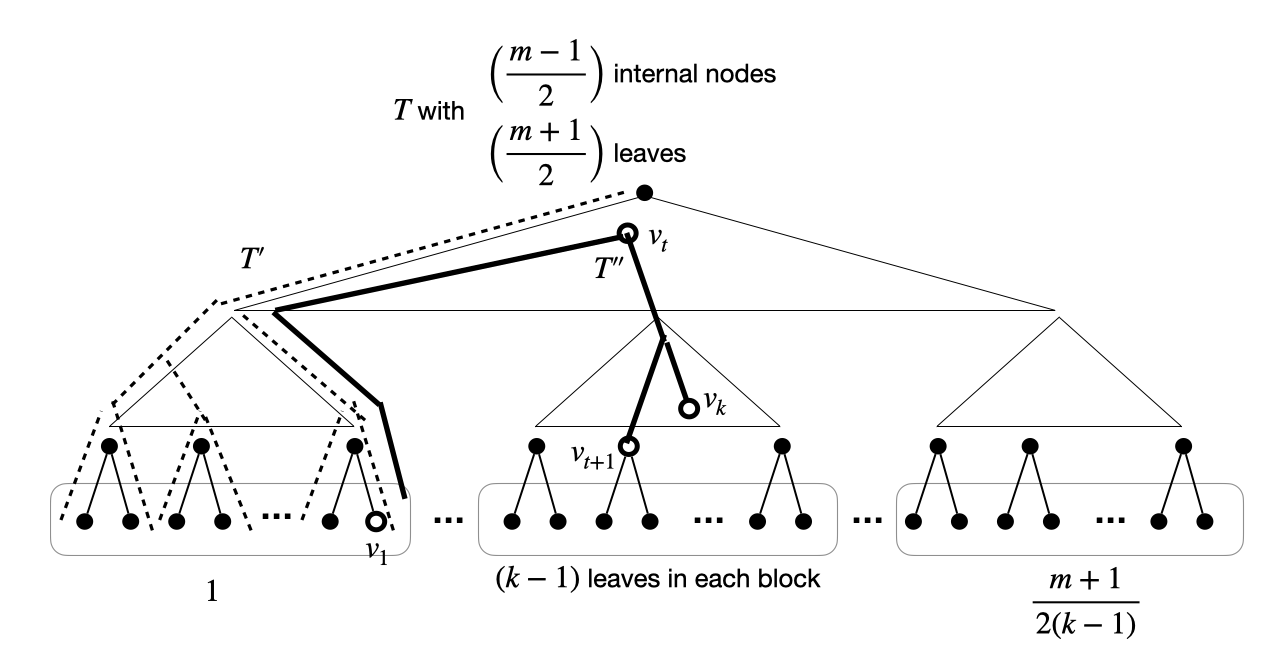}
\caption{The complete rooted binary tree $T$ with $m$ nodes constructed to produce a graph in $\mathcal{G}^c_k$. $T'$ is one of the sub-trees of $T$ with exactly $k$ leaves corresponding to a dependent vertex in $U'$. $T''$ is the sub-tree corresponding to a dependent vertex in $V(H)\backslash(U \cup U')$ constructed from $\{v_1,\ldots,v_t,\ldots,v_k\}$ which are the $k$ selected nodes of $T$.}
\label{fig:example}
\end{figure}

\noindent
Thus, $H$ is defined by $\mathcal{J}=\{J_1,\ldots,J_{n-m-F}\}$ where each $J_i$ defines a sub-tree $T_i$. For an example construction refer Figure~\ref{fig:example}. From the construction above we have the following lemma.
\begin{lemma}
    \label{lem:Gprimelsuperset1}
    $\mathcal{G}^c_k \subseteq \mathcal{G}'_k$.
\end{lemma}
\begin{proof}
From the construction given above, any graph $G \in \mathcal{G}^c_k$ has a fixed $U \subset V(G)$ where $|U|=m$ that is colored using colors $\{1,\ldots,m\}$ such that there exists a tree model of $G$ with:
    \begin{enumerate}
        \itemsep0em
        \item vertices in $U$ corresponding to single node sub-trees, 
        \item sub-trees corresponding to vertices in $U$ are colored using colors $\{1,\ldots,m\}$,  
        \item sub-trees corresponding to vertices in $V(G) \backslash U$ with at most $k$ leaves are uncolored, and
        \item sub-trees corresponding to vertices in $U'$ have $k$ leaves.
    \end{enumerate}
    Since $\mathcal{G}^c_k$ is a special class of partially colored chordal graphs with leafage $k$, $\mathcal{G}^c_k \subset \mathcal{G}'_k$. 
\end{proof}

\noindent
{\bf Computing $|\mathcal{G}^c_k|$}. In order to compute $|\mathcal{G}^c_k|$ and use Proposition~\ref{prop:partialcoloring}, we first note the following from Lemma~\ref{lem:Gprimelsuperset1}.

\begin{proposition}
    \label{lem:Gprimelb}
    $\log |\mathcal{G}'_k| \ge \log |\mathcal{G}^c_k|$.
\end{proposition}

\noindent
Let $K$ denote the set of all possible $\mathcal{J}$'s. We have the following useful proposition and lemma.

\begin{proposition}
    \label{prop:distsubtree1}
    For $1 \le i, i' \le n-m-F, 1 \le j \le m$, if $J_i \ne J_{i'}$ then wlog there exists $a_j \in U$ such that $a_j \in J_i$ and $a_j \notin J_{i'}$.
\end{proposition}


\begin{lemma}
    \label{lem:distinct_graphs}
    Let $\mathcal{J}, \mathcal{J}' \in K$ where $\mathcal{J}=\{J_1,\ldots,J_{n-m-F}\}$ and $\mathcal{J}'=\{J'_1,\ldots,J'_{n-m-F}\}$ such that for $1 \le s \le n-m-F, J_s \ne J'_s$. Also, let $H_1$ and $H_2$ be the graphs generated from $\mathcal{J}$ and $\mathcal{J}'$, respectively, and $u$ be the vertex corresponding to $s$. Then, $N_{H_1}(u) \cap U \ne N_{H_2}(u) \cap U$ where $N_{H_1}(u)$ and $N_{H_2}(u)$ are the neighbours of $u$ in $H_1$ and $H_2$, respectively.
\end{lemma}
\begin{proof}
    From Proposition~\ref{prop:distsubtree1}, we know that for $1 \le j \le m$, if $a_j \in J_s$ and 
    $a_j \notin J'_s$ then $a_j \in N_{H_1}(u) \cap U$ and $N_{H_2}(u) \cap U$. Hence, the lemma.
\end{proof}

\noindent
The following is the central lemma used to obtain the lower bound. For $1 \le s \le n-m-F$, let $\mathcal{J}(H)$ denote the $\mathcal{J} \in K$ that produces the graph $H \in \mathcal{G}^c_k$. 

\begin{lemma}
\label{lem:distinct}
    Let $\langle H_1,U,g\rangle, \langle H_2,U,g\rangle$ be constructed from $\mathcal{J},\mathcal{J}' \in K$. Then $\langle H_1,U,g\rangle$ and $\langle H_2,U,g\rangle$ are same if and only if $\mathcal{J} = \mathcal{J}'$.
\end{lemma}
\begin{proof}
    We prove both directions as follows:
    \begin{enumerate}
        \item If $\langle H_1, U, g \rangle=\langle H_2, U, g \rangle$ then $\mathcal{J}=\mathcal{J}'$. We prove its contrapositive, that is, if $\mathcal{J} \ne \mathcal{J}'$ then $\langle H_1, U, g \rangle \ne \langle H_2, U, g \rangle$. If $\mathcal{J} \ne \mathcal{J}'$ then for $1 \le s \le n-m-F$ there exists $J_s \ne J'_s$. Let $V$ denote the set of vertices in $H_1$ and $H_2$. By Lemma~\ref{lem:distinct_graphs}, there exists $u \in V\backslash U$ corresponding to $s$ such that $N_{H_1}(u) \cap U \ne N_{H_2}(u) \cap U$ where $N_{H_1}(u)$ and $N_{H_2}(u)$ are the neighbours of $u$ in $H_1$ and $H_2,$ respectively. Thus, $E(H_1) \ne E(H_2)$ and so by definition $\langle H_1, U, g \rangle \ne \langle H_2, U, g \rangle$.
        \item If $\mathcal{J} = \mathcal{J}'$ then $\langle H_1, U, g \rangle$ and $\langle H_2, U, g \rangle$ are same by construction.
    \end{enumerate}    
\end{proof}

\noindent
In order to obtain $|\mathcal{G}^c_{t,d}|$ we prove the following lemma first. 

\begin{lemma}
    \label{lem:Gc_size}
    $|\mathcal{G}^c_k| = |K|$.
\end{lemma}
\begin{proof}
    The proof establishes a bijection between $\mathcal{G}^c_{t,d}$ and $K$ as follows. 
    \begin{itemize}
        \itemsep0em
        \item Every graph of $\mathcal{G}^c_{t,d}$ is produced by some element of $K$. Every $\mathcal{J} \in K$ produces a $G \in \mathcal{G}^c_{t,d}$. Since every set of nodes obtained from $\mathcal{J}$ describes some $(n-m-F)$ sub-trees with at most $k$ leaves each such that they represent a graph this is true.
        \item For every $\mathcal{J} \in K$, a distinct graph of $\mathcal{G}^c_{t,d}$ is produced. As per Lemma~\ref{lem:distinct}, if $\mathcal{J}, \mathcal{J}' \in K$ are different then the graphs $H,H' \in \mathcal{G}^c_{t,d}$ produced from them are different.
    \end{itemize}
    Thus, $|\mathcal{G}^c_{t,d}| = |K|$.    
\end{proof}

\noindent
The following is an important lemma.
\begin{lemma}
    \label{lem:G_k}
    $\log |\mathcal{G}^c_k| \ge kn \log n - kn \log k$.
\end{lemma}
\begin{proof}
    From Lemma~\ref{lem:Gc_size}, we know that $|\mathcal{G}^c_k|=|K|$. So we count number of $\mathcal{J}$'s that can be obtained. For $1 \le t \le n-m-F$, the $k$ nodes in $J_t$ can be selected in ${m \choose k}$ ways. For $n-m-F$ vertices it can be done in ${m \choose k}^{n-m-F}$ ways, that is, $|K|={m \choose k}^{n-m-F}$. Using the relation that ${m \choose k}$ equals $\big(\frac{m}{k}\big)^k$ asymptotically for very large $m$ compared to $k$ we have ${m \choose k}^{n-m-F} > \big( \frac{m}{k}\big)^{k(n-m-F)}$, that is, there are at least $\big( \frac{m}{k}\big)^{k(n-m-F)}$ distinct $\mathcal{J}$ that can be constructed. Thus, $|\mathcal{G}^c_k| = |K| > \big( \frac{m}{k}\big)^{k(n-m-F)}$. Then, for fixed $\epsilon>1$ and $m=n/\epsilon$, we have $|\mathcal{G}^c_k| \ge \big( \frac{n}{\epsilon k}\big)^{\big(1-\frac{1}{\epsilon}-\frac{1}{2\epsilon (k-1)}\big)kn}$ when $n$ tends to infinity. Taking logarithm we have,
    $$
       \log |\mathcal{G}^c_k| \ge k\Big( 1-\frac{1}{\epsilon}-\frac{1}{2\epsilon (k-1)}\Big)n\log n - k\Big( 1-\frac{1}{\epsilon}-\frac{1}{2\epsilon (k-1)}\Big)n\log k 
    $$
    As $n$ tends to infinity,
    $$
         \log |\mathcal{G}^c_k| \ge kn \log n - kn \log k.
    $$
\end{proof}

\kleafagelb*
\begin{proof}
    Applying  Proposition~\ref{lem:Gprimelb} to Lemma~\ref{lem:colorGprime1} we have,
    \begin{align*}
        \label{eq:eq1}
        \log |\mathcal{G}_k| &\ge \log |\mathcal{G}'_k| - n \log n + (n-m) \log (n-m) + 2n - O(\log n)\\
        &\ge \log |\mathcal{G}^c_k| - n \log n - O(\log n)
    \end{align*}
    Applying Lemma~\ref{lem:G_k} to the above expression, we have,
    \begin{align*}
        \log |\mathcal{G}_k| &\ge (k-1)n \log n -kn \log k - O(\log n)
    \end{align*}
\end{proof}

\noindent
We have the following proposition.
\begin{proposition}
    \label{prop:support}
    For $k>0$, in $o(n^c), c >0$ and sufficiently large $n$, $\log |\mathcal{G}_k| \ge (k-1)n \log n$.
\end{proposition}

\section{Succinct Data Structure}
\label{sec:succrep}
The high-level procedure used to obtain the succinct data structure for $k-$vertex leafage chordal graph $G$ given as $(T,\{T_1,\ldots,T_n\})$ is as follows:
\begin{enumerate}
    \itemsep0em
    \item From the tree $T$, obtain path graph $H$ with $(k-1)n$ vertices by decomposing each sub-tree $T_i, 1 \le i \le n$, into at most $(k-1)$ paths such that each path created corresponds to a vertex of $H$. The input $(T,\{T_1,\ldots,T_n\})$ is decomposed into paths and we get $(T,\{\mathcal{P}_1,\ldots,\mathcal{P}_n\})$ where $\mathcal{P}_i=\mathcal{P}'_i \cup \mathcal{P}''_i, 1 \le i \le n,$ where $\mathcal{P}'_i=\{P_1^i,\ldots,P_{k/2}^i\}$ and $\mathcal{P}''_i=\{P_{k/2+1}^i,\ldots,P_{k-1}^i\}$ such that if we have $T$ and $\mathcal{P}'_i$ we can compute $\mathcal{P}''_i$.
    \item Out of the $(k-1)n$ paths, we store $kn/2$ paths of $H$ and  tree $T$ using the data structure for path graphs given in~\cite{BCNS2024}. The rest of the $(k/2-1)n$ paths are computed from the  tree $T$ and the stored $kn/2$ paths of $H$. In other words, for each $T_i, \mathcal{P}'_i$ is stored and $\mathcal{P}''_i$ is computed from $\mathcal{P}'_i$ and $T$.
\end{enumerate}

\noindent
The succinct representation for $G$ consists of the following contents:
\begin{enumerate}
    \itemsep0em
    \item $(T, \mathcal{P}'_1 \cup \mathcal{P}'_2 \cup \ldots \mathcal{P}'_n)$ is stored using the method for storing path graphs given in~\cite{BCNS2024}. This data structure stores the tree $T$ using the data structure of Lemma~\ref{lem:2ntree}.
    \item for $1 \le j \le k/2$, the mapping from indices of $P_j^i \in \mathcal{P}'_i$ to the indices of paths in the data structure of~\cite{BCNS2024} that stores $(T,\mathcal{P}'_1 \cup \mathcal{P}'_2 \cup \ldots \cup \mathcal{P}'_n)$ as mentioned in the above step.
\end{enumerate}

\noindent
The following lemma is important before getting into the details of the data structure.
\begin{lemma}
    \label{lem:odd2even}
    Consider clique tree $T$ of $k-$vertex leafage chordal graph $G$ such that $k$ is odd. Then there exists a tree model, denoted $T'$, for $G$ with at most $3n$ nodes such that for $1 \le i \le n, T'_i$ is the sub-tree in $T'$ corresponding to $v_i \in V(G)$ and number of leaves of $T_i$ is even.
\end{lemma}
\begin{proof}
To get $T'$ from $T$, pair the leaves of $T_i$ with odd number of leaves. Since $k$ is odd we will have one isolated leaf. Let this leaf be $a \in V(T_i)$. Add edges $\{a_1,a\},\{a_2,a\}$ to $E(T_i)$ to make number of leaves of $T_i$ even. $T$ can contain at most $n$ nodes and there can be at most $2n$ nodes added to make the number of leaves of sub-tree even in $T'$.   
\end{proof}

\noindent
In this paper, we only consider even $k \ge 2$ as any $T_i$ with odd number of leaves can be converted to even number of leaves as per Lemma~\ref{lem:odd2even}.  Note that the total increase in number of nodes of $T$ is only constant times $n$. We first present the method to transform a $k-$vertex leafage chordal graph $G$ to path graph $H$ with $(k-1)n$ vertices followed by the construction of the data structure.

\subsection{Transforming $(T,\{T_1,\ldots,T_n\})$ to $(T,\mathcal{P}'_1 \cup \ldots \cup \mathcal{P}'_n)$}
\label{sec:transformation}

The transformation from $(T,\{T_1,\ldots,T_n\})$ to $(T,\mathcal{P}'_1 \cup \ldots \cup \mathcal{P}'_n)$ happens in  two steps as below:
\begin{enumerate}[a.]
    \itemsep0em 
    \item pre-process the  tree $T$ into an ordinal tree, and
    \item decompose each $T_i, 1 \le i \le n$, into $\mathcal{P}'_i \cup \mathcal{P}''_i$.
\end{enumerate}

\noindent
Pre-processing is done using the method as explained in Section 3 of~\cite{BCNS2024}. We describe it here for ease of reading.

\noindent
{\bf Pre-processing the  Tree.} Fix a root node for $T$ and perform heavy path decomposition on it. For $v \in V(T)$ order its children $(w_1,\ldots,w_c)$ such that $\{v,w_1\}$ is a heavy edge. Let the children adjacent to $v$ by light edges $(w_2,\ldots,w_c)$, be ordered arbitrarily. This ordering of children of a node of the  tree makes it an ordinal tree. Label the nodes of this ordinal tree based on the pre-order traversal; see Section 12.1 of~\cite{CLRS} for more details of the pre-order traversal of trees. Labels assigned to nodes in this manner are called the \textit{pre-order labels of the nodes}. Throughout the rest of our paper, this ordinal rooted  tree labeled with pre-order will be referred to as the tree model; see Section 3 of~\cite{BCNS2024} for more details. After pre-processing, $T$ is an ordinal tree on which heavy-path decomposition is performed such that all heavy edges are left aligned and nodes are numbered based on the order in which they are visited in the pre-order traversal  of  $T$. Since $G$ is $k-$vertex leafage chordal graph, the number  of leaves of $T_i$ is at most $k$. We obtain the tree representation of the path graph $H \in \mathcal{G}_k(2,(k-1)n)$, denoted $(T,\mathcal{P}'_1 \cup \mathcal{P}''_1 \cup \ldots \cup \mathcal{P}'_n \cup \mathcal{P}''_n)$, by carefully selecting the $k-1$ paths from sub-tree $T_i, 1\le i \le n$. The function $\texttt{getPaths}$ that does this is explained next.

\noindent
{\bf Function $\texttt{getPaths}$.} Given the sub-tree $T_i$ of $T$ with number of leaves at most $k$, the function returns the set of paths $\mathcal{P}'_i$ such that:
\begin{enumerate}
    \itemsep0em
    \item $|\mathcal{P}'_i| \le k/2$,
    \item $\mathcal{P}'_i$ along with $T$ can uniquely determine $\mathcal{P}''_i$, and
    \item $T_i = \mathcal{P}'_i \cup \mathcal{P}''_i$.
\end{enumerate}
The function can be implemented as follows. Let the  leaves of $T_i$ labeled based on the order in which nodes are visited in the pre-order traversal of the tree, be $\{l_1,\ldots,l_{k_i}\}, k_i \le  k$. For $1 \le j \le k_i/2$, pair the smallest leaf $l_j$ with the largest leaf $l_{k_i-j+1}$ to get path $P_j$. Let the set of paths obtained from $T_i$ be denoted by $\mathcal{P}'_i$.

\noindent
{\bf Ordering Paths in $\mathcal{P}'_i$.} Let $\mathcal{P}'_i=\{P_1^i,\ldots,P_{k_i}^i\}, 1 \le k_i \le k/2$. For $1 \le j < j' \le k_i/2,$ $P_j^i=(a_j,b_j)$ and $P_{j'}^i=(a_{j'},b_{j'})$, $P_j^i \prec_{\mathcal{P}} P_{j'}^i$ if $a_j \le a_{j'} \le b_{j'} \le b_j$. $\prec_{\mathcal{P}}$ is a total order on $\mathcal{P}_i$, since for any two paths $P_j^i, P_{j'}^i \in \mathcal{P}_i$, either $a_j \le a_{j'} \le b_{j'} \le b_j$ or $a_{j'} \le a_j \le b_j \le b_{j'}$.

\noindent
{\bf Relation Between $P_j^i, P_{j+1}^i \in \mathcal{P}'_i$.} For $1 \le j \le k/2-1$, the paths connecting $P_j^i$ and $P_{j+1}^i$ are the paths in $\mathcal{P}''_i$. They are computed using the function $\texttt{connector}$. Let $P_j^i=(a_j,b_j)$ and $P_{j+1}^i=(a_{j+1},b_{j+1})$. We define $\texttt{connector}$ based on the following two conditions:
\begin{enumerate}
    \itemsep0em
    \item $P_j^i$ and $P_{j+1}^i$ are not intersecting: Since $P_j^i \prec_{\mathcal{P}} P_{j+1}^i$ and $P_{j+1}^i$ is contained inside a sub-tree rooted at $\texttt{lca}(a_j,b_j)$, 
    $\texttt{connector}(i,j,j+1)=(\texttt{lca}(a_j,b_j),\texttt{lca}(a_{j+1},b_{j+1}))$ connects $P_j^i$ and $P_{j+1}^i$. 
    \item $P_j^i$ and $P_{j+1}^i$ are intersecting: $\texttt{connector}(i,j,j+1)=\texttt{NULL}$ in this case.
\end{enumerate}

\noindent
We have the following useful lemmas.

\begin{lemma}
\label{lem:Ti}
For $1 \le i \le n$, the following holds:
\begin{enumerate}
    \itemsep0em
    \item $|\mathcal{P}'_i| \le k/2$ and $|\mathcal{P}''_i| \le k/2-1$,
    \item ${\normalfont T_i = \mathcal{P}'_i \cup \mathcal{P}''_i}$ where ${\normalfont \mathcal{P}''_i=\texttt{connector}(i,1,2) \cup \texttt{connector}(i,2,3) \cup \ldots \cup}$ \\ ${\normalfont \texttt{connector}(i,k_i/2-1,k_i/2)}$, and 
    \item $T=\bigcup\limits_{i=1}^n F_i$ where ${\normalfont F_i= \mathcal{P}'_i \cup \mathcal{P}''_i}$.
\end{enumerate}
\end{lemma}
\begin{proof} The proofs are as follows:
\begin{enumerate} 
    \itemsep0em
    \item There are at most $k$ leaves and every leaf is paired with another to get a path. Thus, the number of paths is $k/2$. In the case of odd number of leaves, we get a path with endpoints $(a,b)$ where $a$ and $b$ are the leaf and root of  $T_i$, respectively. Thus, the number of paths in $\mathcal{P}_i$ is $\lceil k/2 \rceil$. 
    \item A tree $T'$ is the union of paths from the leaves to the root. We will prove that there is a path from every endpoint of $\mathcal{P}_i$ to the root of $T_i$. The proof is by induction on number of leaves of $T_i$. As the base case we have for $l=2$, $T_i=P_1$. We can see that the result is true as there are no connectors. Let it be true for some $l'>1$. Consider the path $P_j', j' \le (l'+1)/2$, obtained by pairing the last two available leaves as per the algorithm of $\texttt{getPaths}$. Then, there are two cases to consider.
    \begin{enumerate}
        \itemsep0em
        \item {\bf $P_{j'-1},P_{j'}$ are intersecting}: In this case, adding $P_{j'}$ to $\mathcal{P}_i$ will ensure that $T_i = \mathcal{P}_i \cup \texttt{connector}(i,1,2) \cup \texttt{connector}(i,2,3) \cup \ldots \cup \texttt{connector}(i,k_i/2-1,k_i/2), k_i=l'+1$ where $\texttt{connector}(i,k_i/2-1,k_i/2)=\texttt{NULL}$. Since  $P_{j'-1},P_{j'}$ are intersecting there is a path from the end points of $P_{j'}$ to the root of $T_i$.
        \item {\bf $P_{j'-1},P_{j'}$ are non-intersecting}: In this case, adding $P_{j'}$ to $\mathcal{P}_i$ and $\texttt{connector}(i,k_i/2-1,k_i/2) \ne \texttt{NULL}$ will ensure that $T_i = \mathcal{P}_i \cup \texttt{connector}(i,1,2) \cup \texttt{connector}(i,2,3) \cup \ldots \cup \texttt{connector}(i,k_i/2-1,k_i/2), k_i=(l'+1)/2$. Since $\texttt{connector}(i,k_i/2-1,k_i/2) \ne \texttt{NULL}$ it connects $P_{j'-1}$ and $P_{j'}$ and thus, there is a path from end points of $P_{j'}$ to the root of $T_i$.
    \end{enumerate}
    \item Since, $T=\bigcup\limits_{i=1}^n T_i$ is a property of chordal graphs and $T_i=F_i$ as shown above; see~\ref{thm:chordal} for characterisations of chordal graphs.
\end{enumerate}
\end{proof}

\begin{lemma}
\label{lem:lightconn}
For $1 \le j \le k_i/2$, let $P_j^i \in \mathcal{P}'_i$ such that $P_j^i=(a_j^i,b_j^i)$, $e={\normalfont \texttt{lca}}(a_j^i,b_j^i)$ and $Q_1=(e,b_j^i)$. Also, let ${\normalfont \texttt{connector}(i,j,j+1) \ne \texttt{NULL}}$. Then,
\begin{enumerate}
    \itemsep0em
    \item $(e,c) \in E(Q)$ is a light edge, and
    \item $(e,c') \in E({\normalfont \texttt{connector}}(i,j,j+1))$ is a light edge.
\end{enumerate}
\end{lemma}

\noindent
Finally, we have the following proposition and lemma regarding adjacency and neighbourhood of $v \in V(G)$.

\begin{proposition}
\label{prop:adj}
$\{u,v\} \in E(G)$ if and only if there exists $P \in \mathcal{P}'_u \cup \mathcal{P}''_u, Q \in \mathcal{P}'_v \cup \mathcal{P}''_v$, such that $P \cap Q \ne \phi$.
\end{proposition}

\noindent
For $1 \le j \le k_v, 1 \le j' \le k_v-1$, let $\beta(P_j^v)$ and $\beta(\texttt{connector}(v,j',j'+1))$ represent the paths in $\mathcal{P}'_i \cup \mathcal{P}''_i \cup \ldots \cup \mathcal{P}'_n \cup \mathcal{P}''_n$ intersecting $P_j^v$ and $\texttt{connector}(v,j',j'+1)$, respectively.

\begin{lemma}
\label{lem:beta}
For $1 \le i \le n$, let $\mathcal{P}_i = \mathcal{P}'_i \cup \mathcal{P}''_i$. For $P \in \mathcal{P}_i$ the following holds:
\begin{enumerate}
    \itemsep0em
    \item $|\beta(P)| \le (k-1)d_i$, and
    \item $\sum\limits_{P \in \mathcal{P}_i} |\beta(P)| \le (k-1)^2 d_i$ where $d_i$ is degree of $v_i \in V(G)$.
\end{enumerate}
\end{lemma}
\begin{proof}
By Lemma~\ref{lem:Ti}, the  tree $T_i$ corresponding to $v_i \in V(G)$ is decomposed into set of paths $\mathcal{P}_i = \mathcal{P}'_i \cup \mathcal{P}''_i$. 
\begin{enumerate}
    \itemsep0em
    \item Since each $P \in \mathcal{P}'_i$  can have at most $d_i$ trees corresponding to vertices of $G$ intersecting it and each tree has at most $k-1$ paths in the path graph $H$, we have $|\beta(P)| \le (k-1)d_i$.
    \item $\sum\limits_{P \in \mathcal{P}_i} |\beta(P)| \le (k-1)^2 d_i$, since $|\mathcal{P}_i|=k-1$ and  $|\beta(P)| \le (k-1)d_i$.
\end{enumerate}
\end{proof}

\subsection{Construction}

\noindent
{\bf Index of Paths in $G$ and $H$.}  For $P_1^i,\ldots,P_{k_i}^i \in \mathcal{P}'_i$, we index paths of $G$ and $H$ in the following ways:
\begin{enumerate}
    \itemsep0em
    \item in $G$ paths are ordered $\{P_1^1,\ldots,P_{k_1}^1,\ldots,P_1^n,\ldots,P_{k_n}^n\}$ where $P_1^i, 1 \le i \le n$  is in the increasing order of the starting nodes of $P_1^i$ and for $1 \le j \le k_i, P_j^i$ are ordered based on $\prec_{\mathcal{P}}$, and
    \item in $H$ paths are ordered based on their starting nodes; this is as per the storage scheme followed in~\cite{BCNS2024} for path graphs.
\end{enumerate}

\noindent
We will order the vertices of $G$ based on the order of the starting nodes of the first paths $P_1^i \in \mathcal{P}_i$ of each $T_i$. The data structure for $k-$vertex leafage chordal graphs consists of the following components. We distinguish the case when $k$ is odd or even only when the difference alters the higher order term in the space complexity, else we assume $k$ is even. We will show later that $k$ being odd or even does not impact the space complexity of our data structure.

\noindent
{\bf Path Graph $H$.} The path graph $H$ that we store is $(T, \bigcup\limits_{i=1}^n \mathcal{P}'_i)$ where $|\bigcup\limits_i \mathcal{P}'_i| \le nk/2$ if $k$ is even and $|\bigcup\limits_i \mathcal{P}'_i| \le \big(\frac{k-1}{2}+\frac{1}{2}\big)n$ if $k$ is odd. For $1 \le i \le n$, we do not store $\texttt{connector}(i,1,2),\texttt{connector}(i,2,3),\ldots, \texttt{connector}(i,k_i-1,k_i)$ but compute it when required. In other words, $T_i$ can be computed from $T$ and $\mathcal{P}'_i$. Thus, path graph $H$ can be stored using $\frac{nk}{2} \log n + o(nk \log n)$ bits as per Lemma~\ref{lem:succpg} if $k$ is even and $\big(\frac{k-1}{2}+\frac{1}{2}\big)n \log n + o(nk \log n)$ bits if $k$ is odd. The data structures used in the succinct representation of path graphs as given in Lemma~\ref{lem:succpg} does not require $T$ to be a clique tree of $H$ as long as the number of nodes in $T$ is constant times $|V(H)|$, so $(T, \bigcup\limits_{i=1}^n \mathcal{P}_i)$ is a valid input that represents $H$.

\noindent
{\bf Array $K$.} $K$ is a one dimensional array of length $n$. For $1 \le i \le n, K[i]$ stores $|\mathcal{P}'_i|$, that is, the number of paths that the sub-tree $T_i$ gets decomposed into. $K$ takes at most $n \log k$ bits of space. The following function is supported.
\begin{itemize}
    \itemsep0em
    \item $\texttt{getSize}(i):$ Given $1 \le i \le n,$ the function returns $K[i]$ else 0.
\end{itemize}

\noindent
{\bf Bit-vector $F$.} $F$ is a bit vector of length at most $kn/2$. Let the  index of the first path of each vertex in the order $(P_1^1,\ldots,P_{k_1}^1,\ldots,P_1^n,\ldots,P^n_{k_n})$ be $\{i_1,\ldots,i_n\}$ where $P_j^i \in \mathcal{P}'_i, 1 \le i \le n, 1 \le j \le k_i$. Since the vertices are numbered based on the start node of their first paths, $\{i_1,\ldots,i_n\}$ is an increasing sequence of numbers with maximum value of $kn/2$ and can be stored using the differential encoding scheme of Lemma~\ref{lem:ssn}. This data structure is also denoted by $F$ and takes at most $kn + o(kn)$ bits of space. The following function is supported.
\begin{itemize}
    \itemsep0em
    \item $\texttt{getIndex}(i):$ Given $1 \le i \le n,$ the function returns $\texttt{accessNS}(i,F)$ else 0. $\texttt{accessNS}(i,F)$ returns the value at the $i-$th position in the sequence $\{i_1,\ldots,i_n\}$.
\end{itemize}

\noindent
{\bf Bit-vector $D$.} Consider the order of paths $O=(P_1^1,\ldots,P_{k_1}^1,\ldots,P_1^n,\ldots,P_{k_n}^n)$. For $1 \le j \le kn/2, D[j]=i$ if $O[j]$ is a path corresponding to $u_i \in V(G)$. Since  $D$ is an increasing sequence it can be stored  using the data structure of Lemma~\ref{lem:ssn} taking $O(kn)$ bits. $D$ contains at most $n$  1's and $kn/2$ 0's.
 
\noindent
{\bf Permutation $\pi$.} $\pi$ store the mapping of indices of paths in $H$ to paths in $G$ by storing the indices of paths $(P_2^1,\ldots,P_{k_1}^1,\ldots,P_2^n,\ldots,P^n_{k_n})$ in $G$ in that order.  $\pi$ is stored using the data structure of Lemma~\ref{lem:perm} and takes $\left(\frac{k}{2}-1 \right)n \log n + o(kn \log n)$ bits of space if $k$ is even and $\left(\frac{k-1}{2}-\frac{1}{2}\right)n \log n + o(kn \log n)$ bits if $k$ is odd.

\noindent
{\bf Function $\texttt{alpha}$.} The function $\texttt{alpha}$ takes $1 \le i \le n$, as input and returns the indices of paths in $\mathcal{P}'_i$ corresponding to $u_i \in V(G)$ in the  tree representation of $H$. Function returns $\{\pi(p),\ldots,\pi(p+s)\}$ where $p \leftarrow \texttt{getIndex}(i)-i+1$ and $s \leftarrow \texttt{getSize}(i)$. $\texttt{getIndex}(i)-i+1$ is the index of $P_2^i$ in $(P_2^1,\ldots,P_{k_1}^1,\ldots,P_2^n,\ldots,P^n_{k_n})$.

\begin{lemma}
\label{lem:alpha}
Given the index $1\le i \le n$, of $u_i \in V(G)$, ${\normalfont \texttt{alpha}(i)}$ returns the indices in $H$ of paths in $\mathcal{P}'_i$ corresponding to $u_i$ in $O(k_i)$ time.
\end{lemma}

\noindent
{\bf Bit-vector $C$.} $C$ is an array of $n$ bit-vectors each of length $(k/2-1)$. For $1 \le i \le n, 1 \le j \le k/2-1$, $C[i][j]=1$ if $\texttt{connector}(i,j,j+1) \ne \texttt{NULL}$ else $C[i][j]=0$. $C$ takes a total space of $n(k/2-1)$ bits. The following function is supported.
\begin{itemize}
    \itemsep0em
    \item $\texttt{isConnector}(i,j):$ Given $1 \le i \le n, 1 \le j \le k/2-1$, the function returns true if $C[i][j]=1$ else false.
\end{itemize}

\noindent
{\bf Function $\texttt{getCPath}$.} For $1 \le i \le n, 1 \le j \le k_i/2-1$ and $P_j^i=(a_j,b_j), P_{j+1}^i=(a_{j'},b_{j'})$, the function takes paths $P_j^i$ and $P_{j+1}^i$ as input and returns $\texttt{connector}(i,j,j+1)$ if $\texttt{isConnector}(i,j))$ is true else $\texttt{NULL}$. Note that $H$ stored as per Lemma~\ref{lem:succpg} contains the  tree $T$ and this allows us to perform the $\texttt{lca}$ operation.

\begin{lemma}
\label{lem:getCPath}   
For $1\le i \le n, 1 \le j \le k_i/2-1,$ given paths $P_j^i, P_{j+1}^i$ as input, the function ${\normalfont \texttt{getCPath}(P_j^i, P_{j+1}^i)}$ returns ${\normalfont \texttt{connector}(i,j,j+1)}$ in constant time.
\end{lemma}
\begin{proof}
We have $P_j^i=(a_j,b_j), P_{j+1}^i=(a_{j'},b_{j'})$ as input. Using bit-vector $C$ we can check if $\texttt{isConnector}(i,j))$ is true. If true then we can get $\texttt{connector}(i,j,j+1)=(\texttt{lca}(a_j,b_j),\texttt{lca}(a_{j'},b_{j'}))$ in constant time as $\texttt{lca}$ takes constant time. Else, $\texttt{connector}(i,j,j+1)=\texttt{NULL}$. So $\texttt{getCPath}(P_j^i, P_{j+1}^i)$ can be completed in constant time. 
\end{proof}

\begin{lemma}
\label{lem:kleafsuccrep}
For $k>1$ and in $o(n^c), c >0$, there exists a $(k-1)n\log n+o(kn \log n)$-bit succinct data structure for the class of $k-$vertex leafage chordal graphs.
\end{lemma}
\begin{proof}
The data structure for $k-$vertex leafage chordal graphs consists of the following components:
\begin{enumerate}
    \itemsep0em
    \item $H$ takes $\frac{nk}{2} \log n + o(nk \log n)$ bits of space when $k$ is even and $\big(\frac{k-1}{2}+\frac{1}{2}\big)n \log n + o(nk \log n)$ bits if $k$ is odd,
    \item $K$ takes $n\log k$ bits of space,
    \item $F, C,$ and $D$ take $O(kn)$ bits of space each, and
    \item $\pi$ takes $\left(\frac{k}{2}-1 \right) n\log n + o(kn \log n)$ bits of space and $\left(\frac{k-1}{2} -\frac{1}{2}\right)n \log n + o(kn \log n)$ bits if $k$ is odd.
\end{enumerate}
Thus, the total space taken is $(k-1)n \log n + o(kn \log n)$ bits. From Theorem~\ref{thm:kleafagelb} we have,$ \ge (k-1) n\log n \le \log|\mathcal{G}_k| + \frac{3kn}{2} \log k + \frac{3kn}{2} +O(\log n)$. Substituting in $(k-1)n \log + o(kn \log n)$ we have,
\begin{align*}
    (k-1)n \log n + o(kn \log n) &\le \log|\mathcal{G}_k| + \frac{3kn}{2} \log k + \frac{3kn}{2} +O(\log n)\\
    &\le \log |\mathcal{G}_k| + o((k-1)n \log n)
\end{align*}
Using Proposition~\ref{prop:support}, we have,
\begin{align*}
    (k-1)n \log n + o(kn \log n) &\le \log|\mathcal{G}_k|+ o(\log|\mathcal{G}_k|)
\end{align*}
Thus, the data structure is succinct.
\end{proof}

\begin{figure}[ht]
\centering
\includegraphics[width=14cm,height=8.5cm]{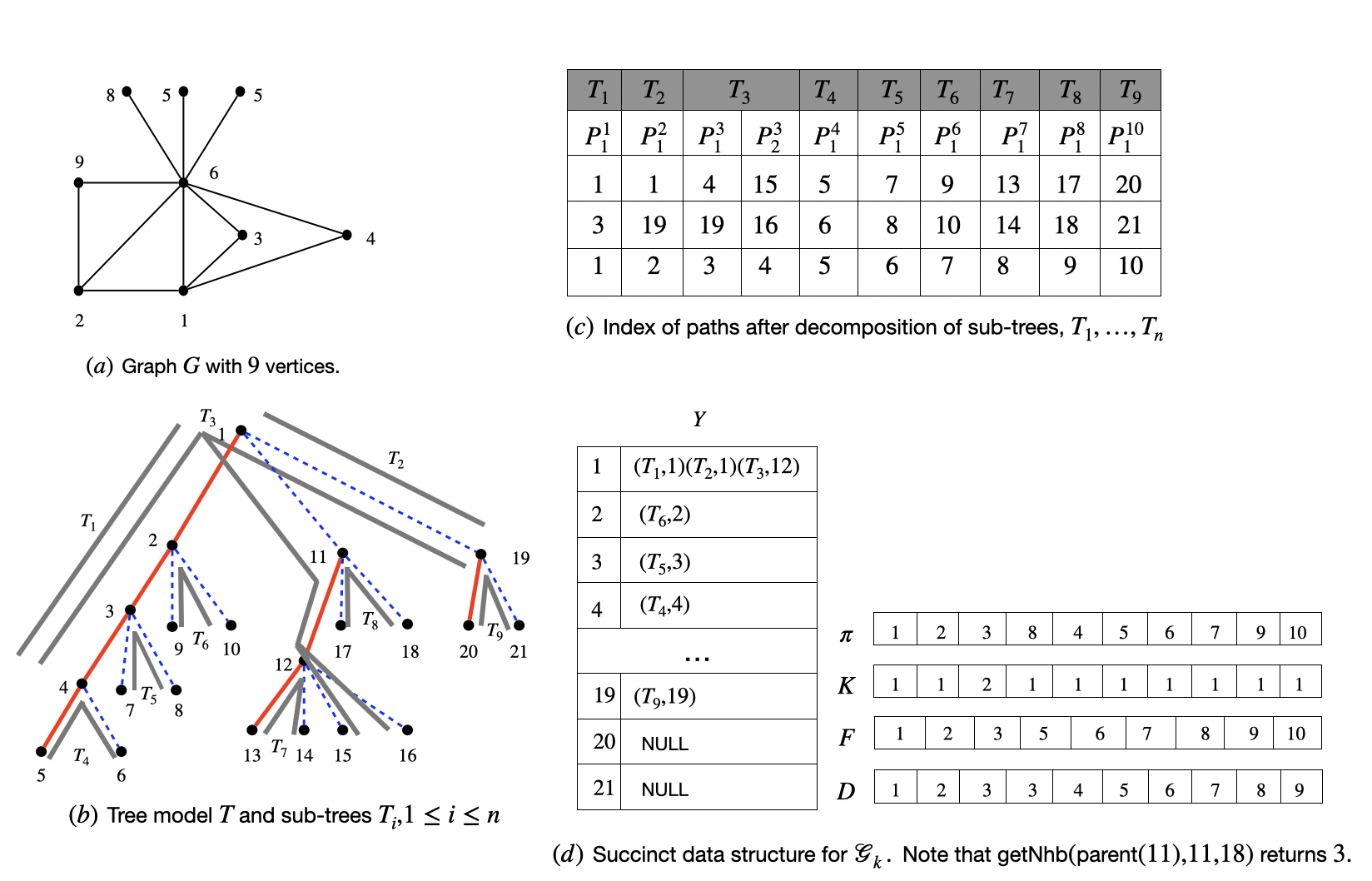}
\caption{(a) An example 4-vertex leafage chordal graph $G$, (b)  tree representation of $G$  after pre-processing, (c) the index of the paths generated from the sub-trees along with their start node (second row), end node (third row), and index (forth row), (d) the components of the succinct data structure for $G$. }
\label{fig:example}
\end{figure}

\subsection{Adjacency and Neighbourhood Queries}
The following lemma is useful.
\begin{lemma}
\label{lem:intersectingconn}
Let $T_i$ and $T_{i'}$ intersect and $P_1^i=(a_1^i,b_1^i), P_1^{i'}=(a_1^{i'},b_1^{i'})$. Wlog,
\begin{enumerate}
    \itemsep0em
    \item if $a_1^i \le a_1^{i'} \le b_1^{i'} \le b_1^i$, then
    \begin{enumerate}[a.]
        \itemsep0em
        \item there exists $P_j^i, 1 \le j \le k_i/2$ such that $P_j^i \cap P_1^{i'} \ne \phi$, or
        \item there exists ${\normalfont\texttt{connector}}(i,j,j+1) \cap P_j^{i'} \ne \phi,  1 \le j \le k_i/2-1$.
    \end{enumerate}
    \item Else, $a_1^i \le a_1^{i'} \le b_1^i \le b_1^{i'}$ or $a_1^{i'} \le a_1^i \le b_1^{i'} \le b_1^i$.
\end{enumerate}

\end{lemma}

\noindent
{\bf Adjacency Query.} Given indices of two vertices $u_i,u_j \in V(G)$ and the succinct representation for the $k-$vertex leafage chordal graphs, the adjacency query returns true if $\{u_i,u_j\} \in E(G)$ else false. The characterisation  of $\{u_i,u_j\} \in E(G)$ in terms of path intersections in $(T,\mathcal{P}_1 \cup \ldots \cup \mathcal{P}_n)$ where $\mathcal{P}_i=\mathcal{P}'_i \cup \mathcal{P}''_i, 1 \le i \le n,$ is given by Proposition~\ref{prop:adj}. Algorithm~\ref{alg:adjacencykleaf} gives an implementation of the adjacency query.

\begin{algorithm}[ht!]
\DontPrintSemicolon
    \SetKwFunction{Fadjacency}{adjacency}
    \SetKwProg{Fn}{Function}{:}{}
    \Fn{\Fadjacency{$i,j$}}{
        $\mathcal{P}'_i \leftarrow \texttt{alpha}(i), \mathcal{P}'_j \leftarrow \texttt{alpha}(j), s \leftarrow \texttt{NULL}$\;
        Let $P_1^i=(a_1^i,b_1^i)$  and $P_1^j=(a_1^j,b_1^j)$ obtained using $\texttt{pathep}(\texttt{getIndex}(i))$ and $\texttt{pathep}(\texttt{getIndex}(j))$, respectively.\;
        Check if $P_1^i$ and $P_1^j$ intersect as follows:\;
        \If{${\normalfont \texttt{adjacencyPG}(H, P_1^i, P_1^j)}$ is true}{
            \KwRet true\;
        }
        Check if endpoints of first path of one falls within the range of the end points of the first path of the other as follows:\;
        \If{$a_1^i \le a_1^j \le b_1^j \le b_1^i$}{$s \leftarrow i$}\; 
        \If{$a_1^j \le a_1^i \le b_1^i \le b_1^j$}{$s \leftarrow j$}\;
        \If{${\normalfont s \ne \texttt{NULL}}$}{
            \If{$s=i$}{
                \ForEach{${\normalfont 1 \le t  \le \texttt{getSize}(i)-1}$}{
                    \If{${\normalfont \texttt{adjacencyPG}(H, \texttt{pathep}(\mathcal{P}'_i[t]),\texttt{pathep}(\mathcal{P}'_j[1]))}$ is true}{\KwRet true\;}
                }
                \ForEach{${\normalfont 2 \le t \le \texttt{getSize}(i)-1}$}{
                    \If{$t=2$ and ${\normalfont \texttt{adjacencyPG}(H, \texttt{getCPath}(\texttt{getIndex}(i),\mathcal{P}'_i[t]),\texttt{pathep}(\mathcal{P}'_j[1]))}$ is true}{
                        \KwRet true\;
                    }
                    \Else{
                        \If{${\normalfont \texttt{adjacencyPG}(H, \texttt{getCPath}(\mathcal{P}'_i[t-1],\mathcal{P}'_i[t]),\texttt{pathep}(\mathcal{P}'_j[1]))}$ is true}{
                            \KwRet true\;
                        }
                    }
                }
            }
            \If{$s=j$}{
                Perform the same steps as done in Line 15 for the case $s=i$ with $i$ and $j$ interchanged.
            }
        }
        \KwRet false\;
    }  
\caption{Given indices $i,j$ as input, the function returns true if $\{u_i,u_j\} \in E(G)$ else false.}    
\label{alg:adjacencykleaf}
\end{algorithm} 

\begin{lemma}
\label{lem:adjacencykleaf}
For $k\in o(n^c), c>0$, the class of $k-$vertex leafage chordal graphs have a $(k-1)n \log n + o(kn \log n)$ bit succinct data structure that supports adjacency query  in $O(k \log n)$ time.    
\end{lemma}
\begin{proof}
As per Lemma~\ref{lem:kleafsuccrep}, $k-$vertex leafage chordal graphs have a $(k-1)n \log n+o(kn \log n)$ bit succinct data structure. Given the succinct representation and indices $i,j$ of vertices in $G$, Algorithm~\ref{alg:adjacencykleaf} checks if $\{u_i,u_j\} \in E(G)$. In Line 1 of Algorithm~\ref{alg:adjacencykleaf}, the paths corresponding to $T_i$ and $T_j$ are obtained using $\texttt{alpha}$ of Lemma~\ref{lem:alpha}  and in Line 2, the endpoints of $P_1^i$ and $P_1^j$ are obtained using $\texttt{pathep}$ of Lemma~\ref{lem:succpg}. The endpoints $(a_t^i,b_t^i)$ of $P_t^i \in \{P_1^i,\ldots,P_{k_i}^i\}$ is such that $a_1^i \le a_t^i \le b_t^i \le b_1^i$; same is true for $j$. The following two cases arise based on Lemma~\ref{lem:intersectingconn}:
\begin{enumerate}
    \itemsep0em
    \item $P_1^i$ and $P_1^j$ are not contained inside each other:  If the sub-trees intersect then, this corresponds to the case (b) of Lemma~\ref{lem:intersectingconn}. In this case, we check if $P_1^i$ and $P_1^j$ intersect or not. In Algorithm~\ref{alg:adjacencykleaf},  this is done in  Lines 4 to 6.
    \item $P_1^j$ is contained inside the sub-tree rooted at $\texttt{lca}(a_1^i,b_1^i)$ or vice versa: If the sub-trees intersect then, this corresponds to the case (a) of Lemma~\ref{lem:intersectingconn}. Wlog let $P_1^j$ be contained inside the sub-tree rooted at $\texttt{lca}(a_1^i,b_1^i)$ as shown in Line 8 of Algorithm~\ref{alg:adjacencykleaf}. In this case we need to check if $P_1^j$ intersects any one  of $P_t^i, 2 \le t \le k$. If $P_1^j$ does not intersect any  one  of $P_t^i$ then no other path of $u_i$ will intersect any path $P_t^i$. This is because the root of the sub-tree $T_j$ is $\texttt{lca}(a_1^j,b_1^j)$. The adjacency check is done using the $\texttt{adjacencyPG}$ function of path graphs data structure of Lemma~\ref{lem:succpg} and takes $O(k \log n)$ time for $k$ paths of $u_i$. This case is handled in Lines 15 to 24.
\end{enumerate}
The function returns false if all the other checks fail to return true. By Lemma~\ref{lem:alpha}, $\texttt{alpha}$ takes $O(k)$ time. By Lemma~\ref{lem:succpg}, $\texttt{pathep}$ and $\texttt{adjacencyPG}$ take $O(\log n)$ time. Also, $\texttt{getSize}$ and $\texttt{getIndex}$ takes constant time. Each of the  loops run at most $k$ times with loop body taking at  most $O(\log n)$ time. Thus, the total time to finish adjacency function is $O(k \log n)$.
\end{proof}

\noindent
{\bf Neighbourhood Query.} Given $v  \in V(G)$ and the succinct representation of $k-$vertex leafage chordal graph $G$, the neighbourhood query returns the neighbours of $v$. We use the following additional data structure.

\noindent
{\bf Array $Y$.} $Y$ is a one dimensional array of length at most $n$ that stores an array at each of its locations. For $1 \le i \le n$, $Y[i]$ stores the array of records where each record is of the form $(r,s)$ such that $1 \le r \le n$ is the index of the sub-tree that has node $u_i \in V(T)$ as the lca of endpoints of $P_1^r$ and for $1 \le s \le n$, node $u_s$ as the lca of endpoints of $P_{k_r}^r$. The records at $Y[i]$ are stored in the increasing order of $s$. The total space taken by $Y$ is $2n \log n$ bits as each tree takes $2 \log n$ bits and there are $n$ trees. The following function is supported.

\noindent
\begin{itemize}
    \item $\texttt{getNhb}(l,l',a)$: Given $1 \le l,l' \le n,$ the function returns the list of trees with index $1 \le j \le n$, such that lca of endpoints of $P_1^j$ is equal to $l$ and lca of endpoints of $P_{k_j}^j$ greater than or equal to $l'$ and less than or equal to $a$ in $O(\log n+d)$ time where $d$ is the number of trees returned. The function can be implemented as follows:
    \begin{enumerate}
        \itemsep0em
        \item Performing binary search on the array $A$ of records at $Y[l]$ to first obtain the range of $s$ values greater than or equal to $l'$. Let this range be $[p_1,p_2]$.
        \item Performing second binary search on $A[p_1,p_2]$ to obtain the range of $s$ values that are less than or equal to $a$ in $A[p_1,p_2]$. Let this range be $[p'_1,p'_2]$.
        \item Return the indices of trees, that is, the $r$ values stored in records $A[p'_1,p'_2]$.
    \end{enumerate}
    The binary search takes $O(\log n)$ time and the tree indices are returned in $d$ time where $d=p'_2-p'_1+1$. Thus, $\texttt{getNhb}$ takes a total time of $O(\log n + d)$.
\end{itemize}

\noindent
Algorithm~\ref{alg:neighbourhoodkleaf} gives an implementation  of the neighbourhood query.

\begin{algorithm}[ht!]
\DontPrintSemicolon
    \SetKwFunction{Fneighbourhood}{neighbourhood}
    \SetKwProg{Fn}{Function}{:}{}
    \Fn{\Fneighbourhood{$i$}}{
        $\mathcal{P}'_i \leftarrow \texttt{alpha}(i), \mathcal{P} \leftarrow \phi$\;
        Add $\texttt{pathep}(\texttt{getIndex}(i))$ to $\mathcal{P}$\;
        \ForEach{${\normalfont j \in \mathcal{P}'_i}$}{
            Add $\texttt{pathep}(j)$ to $\mathcal{P}$\;
        }
        Add $\texttt{getCPath}(\texttt{getIndex}(i),\mathcal{P}'_i[1])$ to $\mathcal{P}$ if it is not $\texttt{NULL}$\;
        \ForEach{${\normalfont 1 \le j \le \texttt{getSize}(i)-1}$}{
            Add $\texttt{getCPath}(\mathcal{P}'_i[j],\mathcal{P}'_i[j+1])$ to $\mathcal{P}$ if it is not $\texttt{NULL}$\;
        }
        $N \leftarrow \phi$\;
        Let $t$ be a bit-vector of length $n$ initialised to 0\;
        \ForEach{$(a,b) \in \mathcal{P}$}{
            $N' \leftarrow \texttt{neighbourhoodPG}(a,b)$\;
        }
        \ForEach{$j' \in N'$}{
            $p \leftarrow \pi^{-1}(j')$ \;
            \If{$t[p] \ne 1$}{
                Add $D[p]$ to $N$\;
                $t[p] \leftarrow 1$ \;
            }
        }

        $P_1^i \leftarrow \texttt{pathep}(\texttt{getIndex}(i))$\;
        $P_2^i \leftarrow \texttt{pathep}(\mathcal{P}_i[1])$\;
        
        Let $P_1^i=(a_1,b_1), P_2^i=(a_2,b_2), l \leftarrow \texttt{lca}(a_1,b_1), l' \leftarrow l$\;
        \If{${\normalfont \texttt{lmost\_child}}({\normalfont \texttt{parent}}(l))=l$}{
            $v \leftarrow {\normalfont \texttt{parent}}({\normalfont \texttt{getHPStartNode}}(l))$\;
        }
        \Else{
            $v \leftarrow {\normalfont \texttt{parent}}(l)$\;
        }
        \While{$v \ne {\normalfont \texttt{NULL}}$}{
            Add $\texttt{getNhb}(v,l',b_1)$ to $N$\;
            $l \leftarrow v$\;
            \If{${\normalfont \texttt{child}}(1,{\normalfont \texttt{parent}}(l))=l$}{
                $v \leftarrow {\normalfont \texttt{parent}}({\normalfont \texttt{getHPStartNode}}(l))$\;
            }
            \Else{
                $v \leftarrow {\normalfont \texttt{parent}}(l)$\;
            }
        }
    }  
\caption{Given index $i$ as input, the function returns the set of vertices adjacent to $u_i \in V(G)$.}    
\label{alg:neighbourhoodkleaf}
\end{algorithm} 

\begin{lemma}
\label{lem:neighbourhoodkleaf}
For $k \in o(n^c), c>0$, the class of $k-$vertex leafage chordal graphs have a $(k-1)n \log n + o(kn \log n)$ bit succinct data structure that supports neighbourhood query for vertex $v_i$ in $O(k^2 d_i \log  n + \log^2 n)$ time using additional $2n\log n$ bits where $d_i$ is the degree of $v_i$.  
\end{lemma}
\begin{proof}
Let $G$ be a $k-$vertex leafage chordal graph and the sub-tree $T_i$ of  tree $T$ of $G$ correspond to $v_i \in V(G)$. Lines 2 to 8 of Algorithm~\ref{alg:neighbourhoodkleaf} decompose $T_i$ into paths $\mathcal{P}=\{P_1^i,\ldots,P_{k_i}^i, \texttt{connector}(i,1,2), \texttt{connector}(i,2,3),\ldots, \texttt{connector}(i,k_i/2-1,k_i/2)\}$ as described in Lemma~\ref{lem:Ti}. In Lines 11 and 12 of Algorithm~\ref{alg:neighbourhoodkleaf} we obtain the neighbours  in $H$ of all paths in $\mathcal{P}$. From Lemma~\ref{lem:beta}, we know that $\sum\limits_{P \in \mathcal{P}} |\beta(P)| \le k(k-1) d_i/2$. Thus, total time taken by Lines 11 and 12 is $O(k^2 d_i \log n)$ as $\texttt{neighbourhoodPG}$ takes $O(d_{v_i} \log n)$ time. The loop at Line 13 takes $O(d_{v_i} f(n))$ time as $\pi^{-1}$ take $O(f(n))$ time as per Lemma~\ref{lem:perm}. All the sub-trees $T_j$ whose connectors intersect $P_1^i=(a_1^i,b_1^i)$ are obtained  by searching for sub-trees with lca as a node in the path connecting the root of $T_i$ to root of $T$ and with a path $P_p^j, 1 \le p \le k_j$ contained inside sub-tree rooted  at $\texttt{lca}(a_1^i,b_1^i)$; these are the only connectors due to Lemma~\ref{lem:intersectingconn}. The loop at Line 25 iterates $O(\log n)$ times till the root is reached. This is because as a consequence of Lemma~\ref{lem:lightconn}, for the connector paths we  only count the light edges of $T$, as the heavy paths are left aligned and they are part of the paths connecting the leaves. Further, as per Lemma  6 of~\cite{BCNS2024}, there are at most $O(\log n)$ light edges in any path of $T$. $\texttt{getNhb}$ takes $O(\log n +d_i)$ time and the loop at Line 21 runs $O(\log n)$ times. So total time taken by the loop is $O(\log^2 n + d_i\log n)$. Since $f(n)=o(\log n)$ as per Lemma~\ref{lem:perm}, the total time taken by Algorithm~\ref{alg:neighbourhoodkleaf} is $O(k^2 d_i \log  n + \log^2 n)$.
\end{proof}

\noindent
Thus, we have the following theorem.
\succinctds*
\begin{proof}
    From Lemma~\ref{lem:kleafsuccrep}, we know that for $k \in o(n/\log n)$, the $(k-1)n \log n + o(kn \log n)$-bit data structure is succinct. From  Lemma~\ref{lem:adjacencykleaf}, we know that the data structure supports adjacency query in time $O(k\log n)$ and from Lemma~\ref{lem:neighbourhoodkleaf}, we know that the neighbourhood query is supported in time $O(k^2 d_i \log  n + \log^2 n)$ using additional $2n\log n$ bits.
\end{proof}

\section{Conclusion}
\label{sec:conclusion}
The parameters, leafage and vertex leafage, are defined for chordal graphs which is a special case of intersection graphs. In comparison, boxicity and interval number allow us to model general graphs as intersection graphs. However, they do not give a tree model like in the case of chordal graphs. We ask the following question: ``\textit{Can leafage and vertex leafage be generalized for any graph?}" The answer is positive if we consider the nice tree-decomposition. For any graph, the nice tree-decomposition allows us to establish a correspondence between vertices of the graph and sub-trees of the tree obtained. As one can see clearly, the parameters leafage and vertex leafage  become applicable to general graphs now. For chordal graphs we know that the clique tree has nodes that correspond to the maximal cliques of the graph. However, we lose such nice properties in the case of nice tree-decomposition. Despite these limitations we think it is worthwhile to generalize these parameters and design a succinct data structure for the more general class thus formed. It is also interesting to note that a unit increase in the leafage parameter increases the number of graphs in the class by $n \log n$ as compared to $2n\log n$ in the case of boxicity or interval number. From Balakrishnan et al.~\cite{GSNK2023Arxiv} we know that for the succinct data structure designed for graphs with bounded boxicity $d>0$ using the succinct data structure for interval graphs an efficient but easy implementation for neighbourhood query is not possible. For bounded leafage parameter where the intersection model is a tree, it will be interesting to see if there exists a simple and efficient implementation of the neighbourhood query.

\bibliography{refs}
\end{document}